\documentclass[showpacs,twocolumn,prl,superscriptaddress]{revtex4-1}
\usepackage{braket}

\usepackage{amsfonts}
\usepackage{amsmath}
\usepackage{graphicx}
\usepackage{enumerate}
\usepackage{amssymb}
\usepackage{amsthm}
\usepackage{subfigure}

\newtheorem{theorem}{Theorem}
\newtheorem{lemma}{Lemma}

\newcommand{\appref}[1]{\hyperref[#1]{{Appendix~\ref*{#1}}}}
\newcommand{\be}{\begin{eqnarray} \begin{aligned}}
\newcommand{\ee}{\end{aligned} \end{eqnarray} }
\newcommand{\benn}{\begin{eqnarray*} \begin{aligned}}
\newcommand{\eenn}{\end{aligned} \end{eqnarray*}}

\begin{document}

\title{Quantifying Coherence and Entanglement via Simple Measurements}

\author{Graeme Smith}
\affiliation{JILA, University of Colorado/NIST, 440 UCB, Boulder, CO 80309, USA}
\affiliation{Center for Theory of Quantum Matter, University of Colorado, Boulder, Colorado 80309, USA}
\affiliation{Department of Physics, University of Colorado, 390 UCB, Boulder, CO 80309, USA}
\author{John A. Smolin}
\affiliation{IBM T.J. Watson Research Center, 1101  Kitchawan Road, Yorktown Heights, NY 10598}
\author{Xiao Yuan}
\affiliation{Center for Quantum Information, Institute for Interdisciplinary Information Sciences, Tsinghua University, Beijing, 100084 China}
\author{Qi Zhao}
\affiliation{Center for Quantum Information, Institute for Interdisciplinary Information Sciences, Tsinghua University, Beijing, 100084 China}
\author{Davide Girolami}
\affiliation{Department of Atomic and Laser Physics, University of Oxford, Parks Road, Oxford OX1 3PU}
\author{Xiongfeng Ma}
\affiliation{Center for Quantum Information, Institute for Interdisciplinary Information Sciences, Tsinghua University, Beijing, 100084 China}

\begin{abstract}
  Coherence and entanglement are fundamental properties of quantum systems, promising to power the near future quantum computers, sensors and simulators. Yet, their experimental detection is challenging, usually requiring full reconstruction of the  system state. We show that one can extract quantitative bounds  to the relative entropy of coherence and the coherent information, coherence and entanglement quantifiers respectively,  by a limited number of purity measurements. The scheme is readily implementable with current technology to verify quantum computations in large scale registers, without carrying out expensive state tomography.
\end{abstract}

\maketitle
The superposition principle is one of the pillars of quantum
mechanics. Coherent superpositions of multipartite states can yield
entanglement, a property of such states that are nonfactorizable.
Entanglement has been extensively investigated \cite{Horodecki09},
having being identified as a key property since the pioneering studies
in quantum communication and cryptography protocols, and for quantum
computational speed-up \cite{dense,Shor97,Ekert91,telep}.
Entanglement is also crucial in quantum condensed matter theory,
because it underpins fundamental properties of many-body systems such
as critical behaviors \cite{vedral} as well as improved metrology
beyond the standard quantum limit \cite{metrorev}. Once confined to
suggestive thought experiments, highly coherent quantum systems are
nowadays observed and manipulated in the laboratory.

The quantification of entanglement is thus of great interest.  Except
for the case of bipartite pure states~\footnote{The von Neumann entropy of the
reduced density matrix of either subsystem is an entanglement measure for bipartite pure states} quantifying entanglement is complicated.
For mixed states, there are many different measures (see
\cite{Measures} for a review).  Depending on one's purpose, one of
more of these may be appropriate. They include the distillable
entanglement \cite{BDSW}, the entanglement of formation \cite{BDSW},
the entanglement cost \cite{hayden2001asymptotic}, the concurrence
\cite{wootters}, and the log negativity \cite{vidal2002computable}.
The most informative entanglement measures have some operational
meaning.  For example, the distillable entanglement is how much pure
entanglement can be extracted from the quantum state in question,
while the entanglement cost is how much pure entanglement is needed to
create a state asymptotically.  Because of their asymptotic nature,
these can be extremely difficult to calculate.  Indeed, there seems to
be a tradeoff between operational meaning and computational
accessibility.  The log negativity, for example, is simple to
calculate given a full description of a quantum state, its density
matrix, but lacks operational interpretation.

In general, it is easier to detect, or certify, entanglement than to
quantify it.  A heavily employed solution is to verify quantum
non-locality by violation of Bell inequalities, which is observed in a
subset of the entangled states \cite{entdet}. Another possibility is
to design and measure witnesses \cite{terhal}.  These are observables
which verify the presence of entanglement whenever their value is
above (or below) a threshold. For example, entanglement-enhanced
precision in phase estimation protocols is certified by super-linear
scaling of the quantum Fisher information
\cite{oberthaler,zoller,hefei}.  These customary detection methods
usually certify nonseparability but offer no quantitative information
about the usefulness of the entanglement present, and often fail to
detect the entanglement at all.  In this work we seek to quantify
entanglement rather than merely detect it, without requiring extraordinarily
difficult experiments.

This most direct way to estimate the entanglement of a state is 
by full tomographic reconstruction. This
task is computationally and experimentally challenging, scaling
exponentially with the number of qubits.  It is therefore desirable to
have a way to quantify the entanglement present in a system without
doing full tomography.  In Ref.~\cite{alves2004multipartite}, it was proposed how to directly
confirm the presence of entanglement by interfering two identical
copies of a state and extracting a quantity $\gamma=\mathrm{Tr}
\rho^2$ which is known as the \emph{purity} of a quantum state.
In Ref.~\cite{Experiment}, such an experiment was carried out on ultracold
bosonic atoms in an optical lattice. Using the methods developed
below, it is possible to not only to verify that the systems were
entangled, but quantify the amount of entanglement present.

\begin{figure}
\subfigure[] {
 \includegraphics[width=1.5 in]{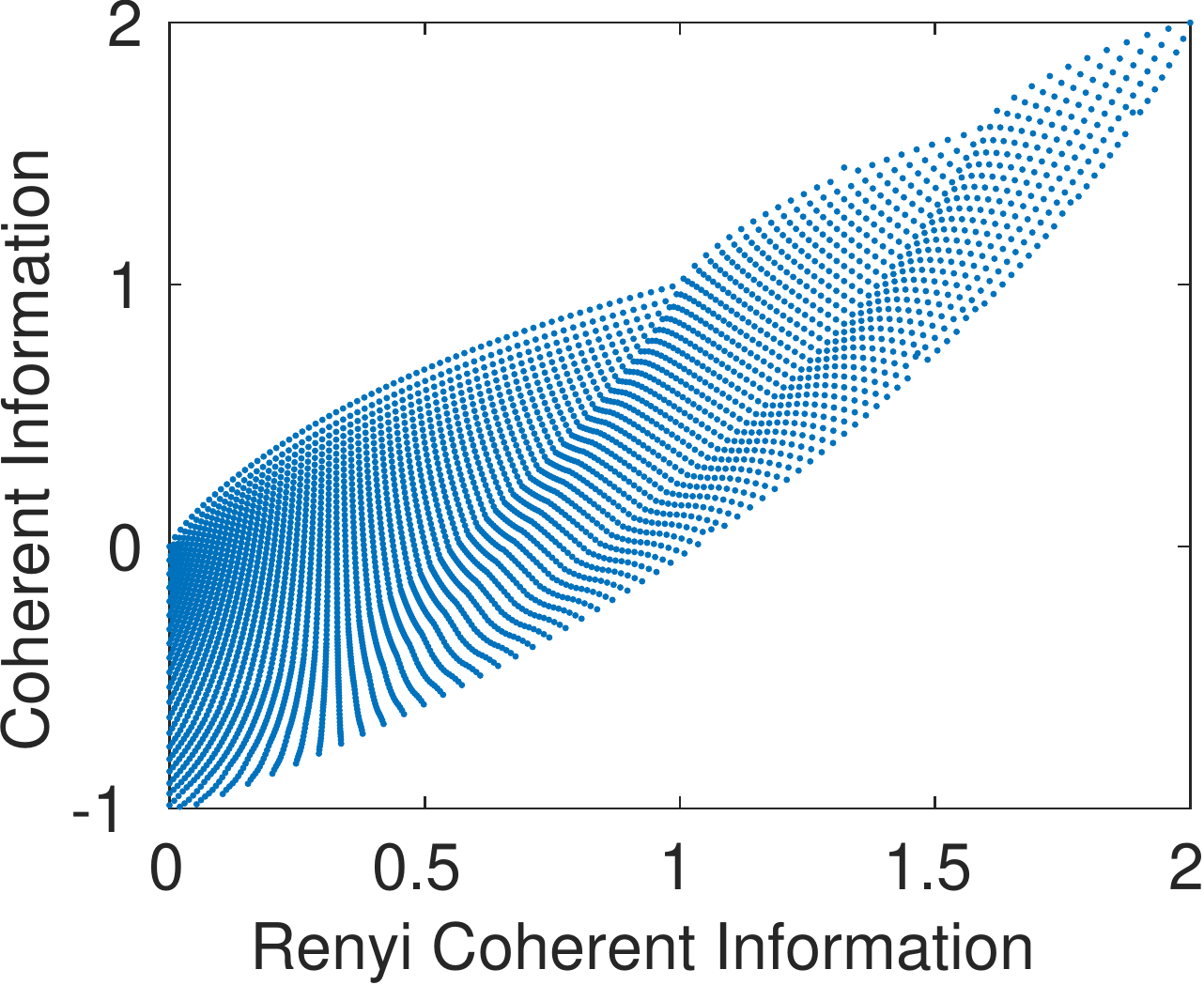}
}
\subfigure[] {
  \includegraphics[width=1.5 in]{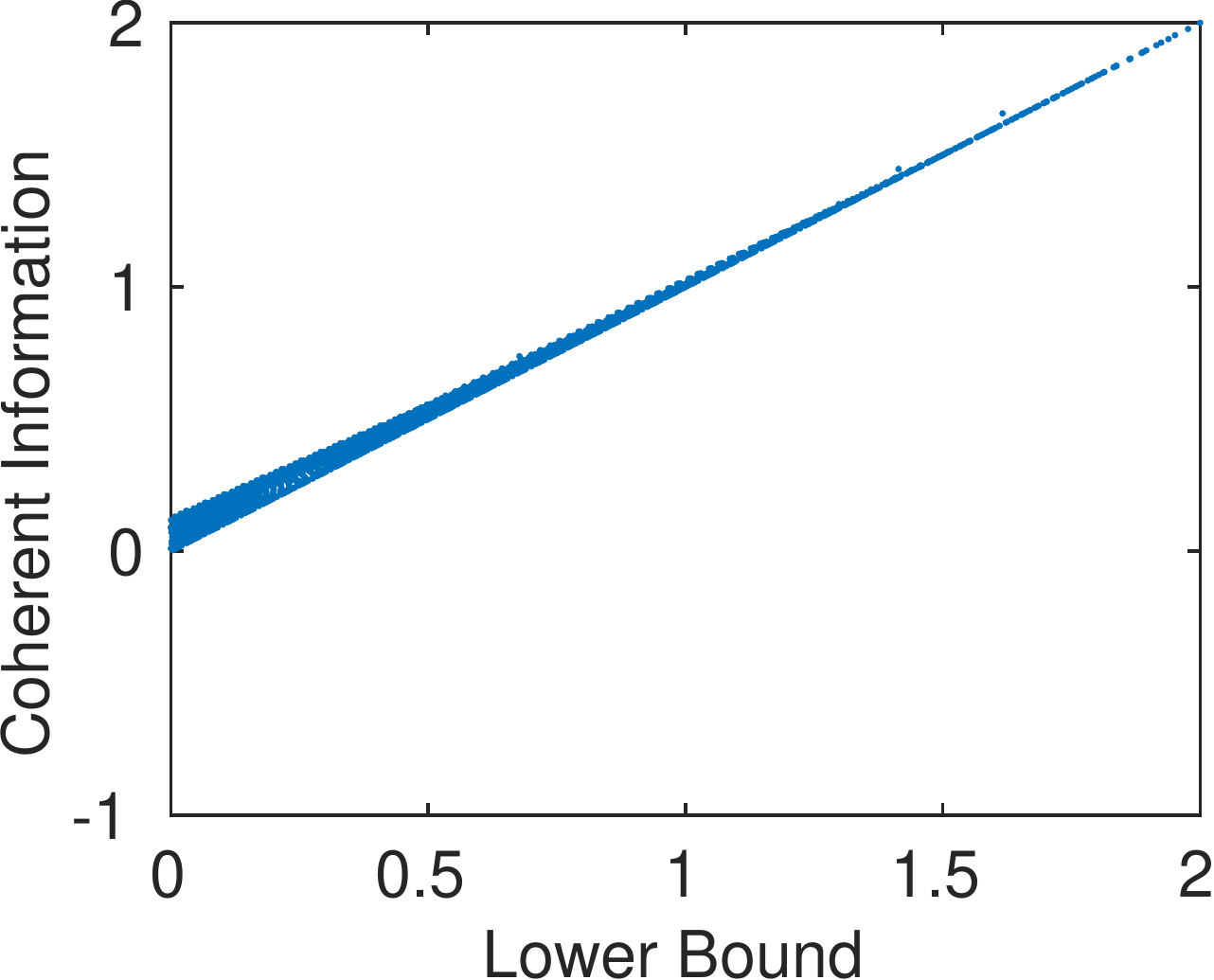}
}
\caption{Comparison for dimension $d=4$. a) Minimum coherent
  information vs. Renyi coherent information. b) Minimum coherent
  information vs. Eq.(\ref{thebound}). Each point displayed
  corresponds to values of the axes related to some pair of
  purity values $\gamma_B,\gamma_{AB}$.  The reason there can be more
  than one value of the coherent information for each $x$ value is
  that there are different combinations of global and local purities
  that lead to the same value of the $x$-axis but different coherent
  information.  We interpret part a) as showing that there are states
  that are definitely entangled because the Renyi coherent information
  is positive but which have little or even negative coherent
  information.  Part b) shows that if you have measured global and
  local purities, then our lower bound (\ref{thebound}) is fairly
  tight--if you have no other information, then there are states with
  just that little coherent information.
    \label{bipartitefig}}
\end{figure}

By measuring the purity of a bipartite system $\gamma_{AB}$, as well as the purity of
(at least) one of the subsystems, one certifies that the systems $A,B$ are entangled whenever
the global purity is lower than the subsystem purity.  This can be expressed
in terms of the Renyi coherent information 
\begin{equation}
S_R=S_2(\rho_B)-S_2(\rho_{AB}),
\end{equation}
as  $S_R > 0$.  Here $\rho_{AB}$ is the density matrix of the entire system and $\rho_B$ is
the reduced density matrix on system $B$, while
\begin{equation}
S_2(\rho)= -\log \mathrm{Tr} \rho^2 = -\log \gamma\ .
\end{equation}
Positivity of $S_R$ implies that the system is entangled \cite{Horodecki09}.
Unfortunately, this Renyi quantity, while able to certify entanglement
(it is an \emph{entanglement witness} \cite{terhal}), does not
quantify it.

While Renyi entropies often play a similar role, the most useful
entanglement quantities are usually given in terms of von Neumann
entropy $S(\rho)=-\mathrm{Tr} \rho \log \rho$.  Then, a more operational
quantity is the \emph{coherent information}, which is defined like
$S_R$, except in terms of von Neumann entropy instead of Renyi:
\begin{equation}
I(A\rangle B)=S(\rho_B)-S(\rho_{AB})\ .  
\end{equation}
Coherent information characterizes the degree to which
error-correction can maintain coherence in the system.
As can be seen in Fig. \ref{bipartitefig}a, there are many states with
high coherent information but very low Renyi coherent information, and
vice versa.

The coherent information is harder to measure than the Renyi version,
because to calculate von Neumann entropies one needs to know the
eigenvalues of a system, not just its purity.  However, one can obtain
quantitative upper and lower bounds on the von Neumann entropy in
terms of the global and marginal purities using the method of Lagrange
multipliers (see the Appendix for details).

\begin{figure}
\includegraphics[width=2.5 in]{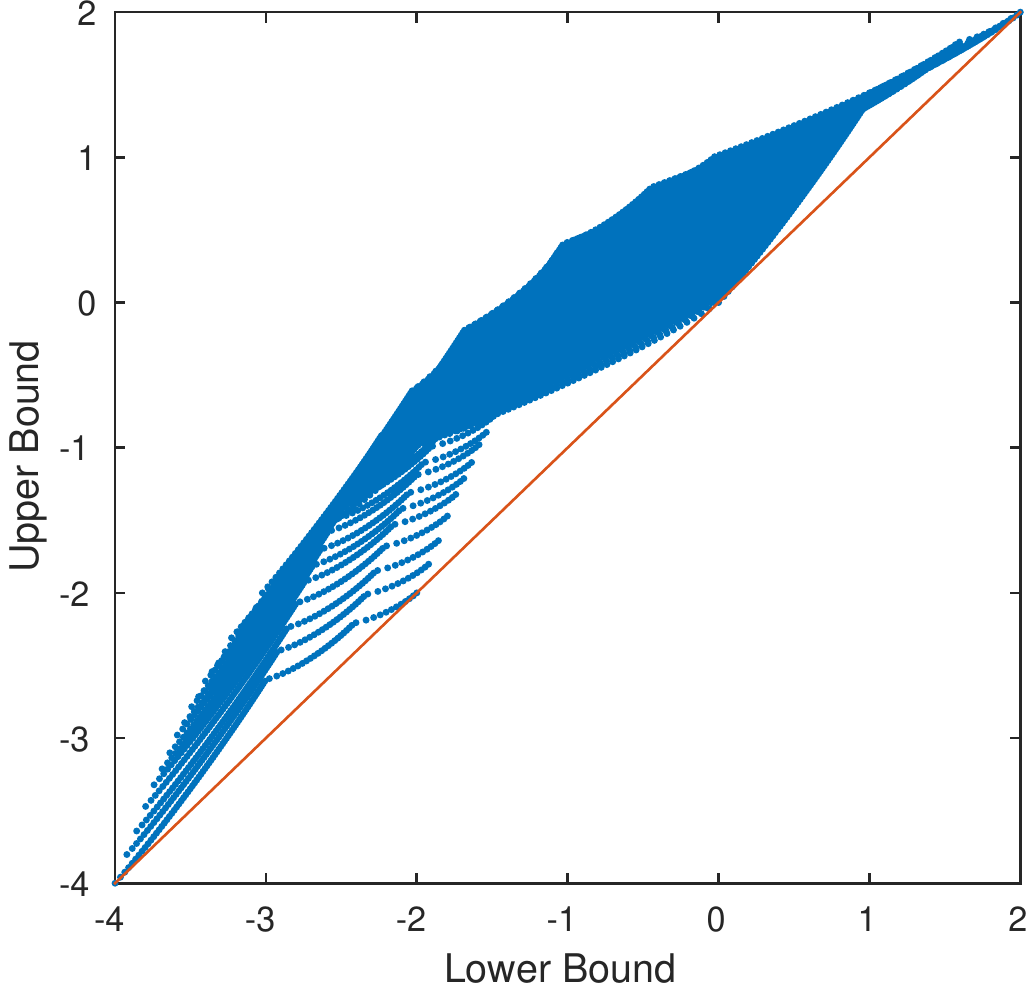}
\caption{Our upper and lower bounds Eq.~(\ref{thebound}) on the coherent
information plotted against each other for dimension $4\times 4$.  Each
point displayed corresponds to the bounds at some pair of purity values
$\gamma_B,\gamma_{AB}$. Again, the
reason there can be more than one value of one bound corresponding
to the other is that there are different combinations of global and
local purities for each value of either bound.
\label{uppervslower}
 }
\end{figure}

\begin{figure}
\includegraphics[width=2.5 in]{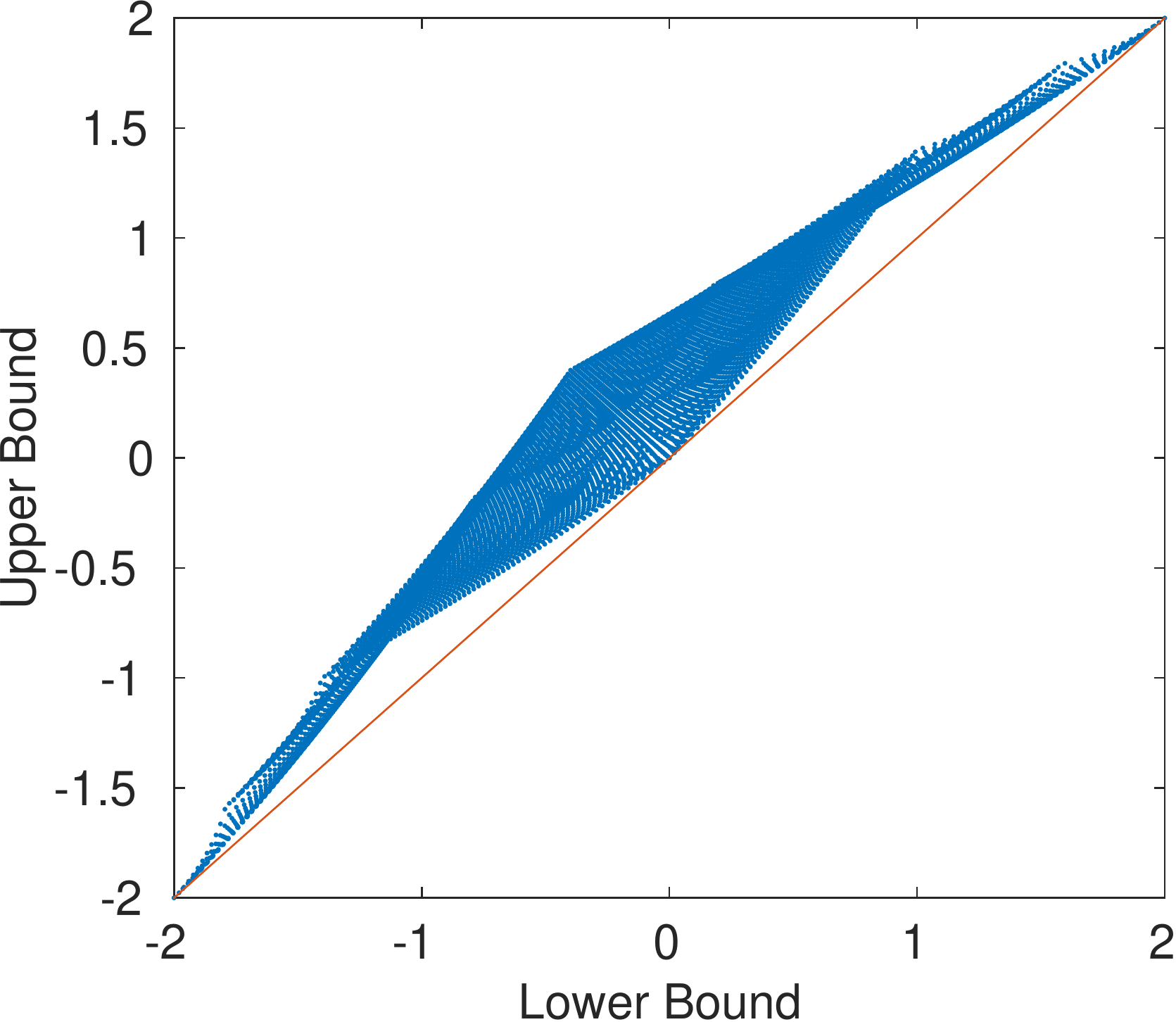}
\caption{Our upper and lower bounds Eq.~(\ref{Eq:Estimationlb}) on the coherence
plotted against each other for dimension $d=4$.  Each
point displayed corresponds the bounds at some pair of purity values
$\gamma_\rho,\gamma_d$. Compare to Fig. \protect\ref{uppervslower}.
\label{uppervslowercoh}
 }
\end{figure}

Given the spectral decomposition $\rho = \sum_{i=1}^dx_{i,\rho}\ket{\psi_i}\bra{\psi_i}, \sum_i x_{i,\rho}=1, \langle\psi_i|\psi_j\rangle=\delta_{ij},  x_{1,\rho}\ge x_{2,\rho} \ge \dots x_{d,\rho}$,  we determine the extreme values of the state entropy $S(\rho)= -\sum_{i=1}^dx_{i,\rho}\log x_{i,\rho}$ at fixed purity $ \gamma_\rho:=\sum_{i=1}^dx_{i,\rho}^2$, where the logarithm is written in base 2 \cite{zyczkowski2003renyi}. 
 The spectrum $\{x_{i,\rho}^{\text{M}}\}$ that maximizes $S(\rho)$ is the following:
\begin{equation}\label{Eq:solutionmaxd}
\begin{aligned}
  x_{1,\rho}^{\text{M}} &= \frac{1}{d}+\sqrt{\frac{d-1}{d}\left(\gamma_\rho-\frac{1}{d}\right)},\\
  x_{2,\rho}^{\text{M}} &=x_{3,\rho}^{\text{M}}=\dots=x_{d,\rho}^{\text{M}} = \frac{1-x_{1,\rho}^{\text{M}}}{d-1}.
\end{aligned}
\end{equation}
The spectrum $\{x_{i,\rho}^{\text{m}}\}$ that minimizes $S(\rho)$ is given by
\begin{equation}\label{Eq:solutionmind}
\begin{aligned}
x_{1,\rho}^{\text{m}}=x^{\text{m}}_{2,\rho}=\dots=x^{\text{m}}_{k_\rho-1,\rho}&=\frac{1-\alpha_\rho}{k_\rho-1},\\
x^{\text{m}}_{k_\rho,\rho} &= \alpha_\rho,\\
x^{\text{m}}_{k_\rho+1,\rho}=\dots=x^{\text{m}}_{d,\rho} &= 0,
\end{aligned}
\end{equation}
where $\alpha_\rho={1}/{k_\rho} - \sqrt{ (1-1/k_\rho)(\gamma_\rho-1/k_\rho)}$ and $k_\rho$ is the integer such that $ \frac{1}{k_\rho} \le  \gamma_\rho \le \frac{1}{k_\rho-1}$.
We can immediately use these results to bound the coherent information as follows:

\emph{Result 1.}--- Given a quantum state $\rho_{AB}\in \mathcal{H}_{d_A}\otimes \mathcal{H}_{d_B}$, its coherent information $I(A\rangle B)$ is bounded as follows:
  \begin{eqnarray}  \label{thebound}
    l_e(\rho_{AB})\leq I(A\rangle B)\leq u_e(\rho_{AB})
 \end{eqnarray}
{\footnotesize
 \begin{eqnarray}
 &&l_e(\rho_{AB})= \nonumber\\
&-&(1-x^{\text{m}}_{k_{\rho_B},\rho_B}) \log x^{\text{m}}_{1,\rho_B}
 -x^{\text{m}}_{k_{\rho_B},\rho_B} \log  x^{\text{m}}_{k_{\rho_B},\rho_B}\nonumber\\
  &+& \frac{(d-1)(1-x^{\text{M}}_{1,\rho_{AB}})}{d} \log \frac{(1-x^{\text{M}}_{1,\rho_{AB}})}{d} + x^{\text{M}}_{1,\rho_{AB}} \log {x^{\text{M}}_{1,\rho_{AB}}},\nonumber\\
&&u_e(\rho_{AB})=\nonumber\\
&& (1-x^{\text{m}}_{k_{\rho_{AB}},\rho_{AB}}) \log x^{\text{m}}_{1,\rho_{AB}} +x^{\text{m}}_{k_{\rho_{AB}},\rho_{AB}} \log  x^{\text{m}}_{k_{\rho_{AB}},\rho_{AB}}\nonumber\\
 &-&  \frac{(d-1)(1-x^{\text{M}}_{1,\rho_B})}{d} \log \frac{(1-x^{\text{M}}_{1,\rho_B})}{d} - x^{\text{M}}_{1,\rho_B} \log x^{\text{M}}_{1,\rho_B}\nonumber
\end{eqnarray}
}
where $\rho_B=\mathrm{Tr}_A (\rho_{AB})$.
These bounds are very good, see Figure \ref{bipartitefig}b and Figure \ref{uppervslower}.

We pause here to note that coherent information is not a full-fledged
entanglement measure \cite{horodecki2000limits}, since it can be zero
or even negative (even for states that \emph{do have}
entanglement). Yet, it characterizes many uses of bipartite
entanglement in quantum computation and communication protocols. The
coherent information measures the capacity for noiseless quantum
communication of a quantum channel between $A$ and $B$, when maximized
over the sets of possible inputs, in the asymptotic limit of an
infinite number of trials \cite{schum,Lloyd,Devetak,Shor}.  It also
quantifies the one-way distillable entanglement
\cite{BDSW,DevetakWinter}. In the quantum state merging protocol
\cite{horodecki2005partial}, it yields the amount of entanglement
which $A$ and $B$ still have after completing the transfer of a
quantum system.

We next study another feature of quantum systems, their \emph{coherence}.
In a way similar to how non-factorizable superpositions of multipartite states, e.g. $\sum_i c_i \ket{ii\ldots i }$, yield entanglement, the quantumness of a system can be identified with the degree of  coherence of its state $\ket{\psi}=\sum_i c_i \ket{i}, \sum_i |c_i|^2=1,$ in a given basis $\{\ket{i}\}$. 
Surprisingly, an information-theoretic characterization of  coherence has been developed only in recent works \cite{cohrev}, in contrast to entanglement which has been extensively investigated \cite{Horodecki09}.

A natural way to quantify the coherence of a state in a  reference basis $\{\ket{1},\ket{2},\dots,\ket{d}\}$ of  a $d$-dimensional Hilbert space $\mathcal{H}_d$
 is by measuring how far it is to the set of incoherent states ${\cal I}$ \cite{BCP,herbut}. The choice of distance function is in principle arbitrary. Yet, an important operational interpretation is enjoyed by the relative entropy of coherence
\begin{equation}\label{CREEQ}
  C_{\mathrm{RE}}(\rho) = \min\limits_{\sigma\in{\cal I}}S(\rho||\sigma)=S(\rho_{d}) - S(\rho),
\end{equation}
where $\rho_{d}=\sum_i \ket{i}\!\!\bra{i}\rho \ket{i}\!\!\bra{i}$ is the state after dephasing in the reference basis.  In other words, coherence is evaluated by how much mixedness a dephasing channel adds to the system state.  The relative entropy of coherence is the distillable coherence of a state \cite{yang}. That is, in the asymptotic limit of infinite system preparations, the maximal rate of extraction of maximally coherent qubit states $1/2\sum_{i,j=0,1}\ket{i}\bra{j}$ by incoherent operations.  This quantity is again easily bounded by purity measurements.

\emph{Result 2} --- The relative entropy of coherence $C_{\mathrm{RE}}(\rho)$ is  bounded as follows:
\begin{eqnarray}\label{Eq:Estimationlb}
  l_c(\rho)\leq C_{\mathrm{RE}}(\rho)\leq u_c(\rho),
\end{eqnarray}
{\footnotesize
\begin{eqnarray}
 &&l_c(\rho)= -(1-x^{\text{m}}_{k_{\rho_d},\rho_d}) \log x^{\text{m}}_{1,\rho_d}-x^{\text{m}}_{k_{\rho_d},\rho_d} \log x^{\text{m}}_{k_{\rho_d},\rho_d}\nonumber\\
  &+& \frac{(d-1)(1-x^{\text{M}}_{1,\rho})}{d} \log\frac{(1-x^{\text{M}}_{1,\rho})}{d} + x^{\text{M}}_{1,\rho} \log{x^{\text{M}}_{1,\rho}},\nonumber\\
 &&u_c(\rho)=  (1-x^{\text{m}}_{k_{\rho},\rho}) \log x^{\text{m}}_{1,\rho} +x^{\text{m}}_{k_{\rho},\rho} \log x^{\text{m}}_{k_{\rho},\rho}\nonumber\\
 &-&  \frac{(d-1)(1-x^{\text{M}}_{1,\rho_d})}{d} \log \frac{(1-x^{\text{M}}_{1,\rho_d})}{d} - x^{\text{M}}_{1,\rho_d} \log x^{\text{M}}_{1,\rho_d}\nonumber.
\end{eqnarray}
}

To summarize, we provided quantitative bounds to coherent information
and relative entropy of coherence in terms of global and marginal
purities.  We now describe the experimental setting required for
measuring state purity.
The purity of a state $\rho$ can be measured on just two copies, $\rho\otimes\rho$---using
precisely the same data as used in \cite{Experiment}. 

This can be done in two ways.
The first method, illustrated in Fig.~\ref{figureexperiment}a is to measure the expectation
value of the swap operator $V$ on $\rho\otimes\rho$, taking advantage
of the identity $\mathrm{Tr}(\rho^2) = \mathrm{Tr}(V
\rho\otimes\rho)$.  This can be accomplished using an ancillary qubit
and a controlled swap~\cite{brun,ekert,filip}.
For multiple qubit systems, implementing a full controlled swap appears
difficult.  However, observing that the swap is factorizable, one can
perform controlled swaps sequentially on the individual corresponding
pairs of qubits from each copy of $\rho$.
We note that the purity of the dephased state is also measurable by applying dephasing before the interaction gate to just one copy of the state, as $\text{Tr}(\rho_d^2)=\text{Tr}(\rho_d\rho)=\text{Tr}(V(\rho_d\otimes\rho))$. 

The measurement can also be accomplished without
ancilla by measuring in the Bell basis.  This is because
\begin{equation}\nonumber
  \mathrm{Tr}( V \rho \otimes \rho) = \mathrm{Tr}\left(({\cal I}\!-\! 2 |\Psi^-\rangle\!\langle \Psi^-|)\rho\otimes\rho\right)
  = 1 - 2 \langle \Psi^-| \rho \otimes \rho | \Psi^-\rangle
\end{equation}
where $|\Psi^+\rangle=\frac{1}{\sqrt{2}}(|01\rangle-|10\rangle)$ is the antisymmetric singlet state.
This second method, ideal for bosonic states, can be achieved by interfering two
copies of a state on a
beamsplitter~\cite{alves2004multipartite,Experiment,Girolami14}.  When a photon is detected
at both output ports of the beamsplitter, the state is projected into the singlet. From
repeated experiments the probability that the output state is the singlet can be determined.
Again, if the state $\rho$ is of many qubits, the beamsplitter can be performed on
individual qubits.  Here there is the drawback that the probability of measuring the output as
all singlets goes down exponentially in the number of qubits in the state, so many measurements
will be needed to evaluate this probability.
This scheme is shown in Fig.~\ref{figureexperiment}b.

Our results Eqs.~(\ref{thebound}) and (\ref{Eq:Estimationlb}) rely on bounding the
von Neumann entropy by quadratic polynomials, i.e. purity. This
represents the leading order term of the von Neumann entropy Taylor
expansion about pure states.  We anticipate that tightened bounds can
be extracted by evaluating the higher order terms
$\mathrm{Tr}(\rho^3), \mathrm{Tr}(\rho^4), \ldots,
\mathrm{Tr}(\rho^{d})$.  Each $k$th-degree polynomial
$\mathrm{Tr}(\rho^k)$ can be estimated by upgrading the scheme in
Fig.~\ref{figureexperiment} to interfere $k$ copies of the state, and
evaluating the shift (generalized swap) operator
$V_k(\phi_1\otimes\phi_2\otimes\ldots\otimes
\phi_k)=\phi_k\otimes\phi_1\otimes\ldots\otimes \phi_{k-1},
\text{Tr}(V_k(\bigotimes_{i=1}^k\rho_i ))=
\text{Tr}(\Pi_{i=1}^k\rho_i), \forall \rho_1, \rho_2, \ldots,
\rho_k$. The protocol would still exponentially outperform full state
reconstruction.

\begin{figure}[t]\centering
  \subfigure[]{\resizebox{1.5 in}{!}{\includegraphics{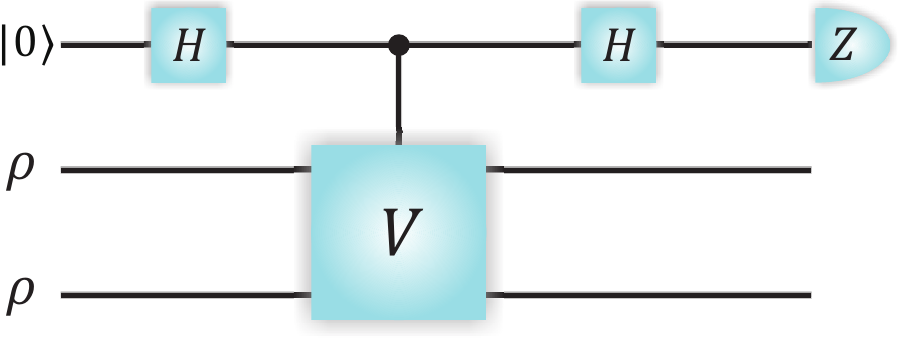}\label{1a}}}
  \subfigure[]{\resizebox{1.5 in}{!}{\includegraphics{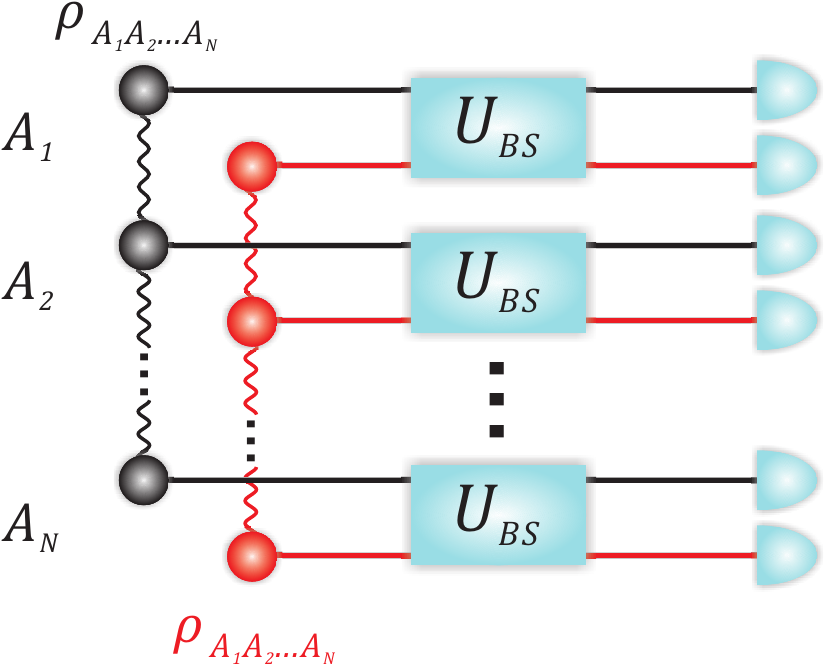}\label{1b}}}
\caption{Purity measurement scheme. a) An ancillary qubit is
  prepared in the state $|0\rangle$ together with two
  identical system copies in the state $\rho$. An interferometric
  network is built up by applying an Hadamard gate $H$, followed by a
  controlled-swap interaction correlating the ancilla with the system
  copies, and a second Hadamard gate. The state purity is encoded in
  the state of the ancilla at the output.
  b) Two identical copies of a state of $n$ qubits $\rho_{A_1A_2\ldots A_n}$ are
  interfered on $n$ beamsplitters. Coincidence observed on all detectors indicates
  projection into the singlet state.
  \label{figureexperiment}}
\end{figure}


One can extend the method proposed here to determine directly measurable bounds to the total correlations in multipartite systems $\{A_i\}$. Let us consider a geometric measure of correlations, the multi-information given by the relative entropy between the state under study and the closest product state,
\begin{equation}{\cal I}(\rho_{A_1,\ldots,A_n}) =\min\limits_{\bigotimes_i\sigma_{A_i}}S(\rho_{A_1,\ldots,A_n}||\bigotimes_i\sigma_{A_i}).
\end{equation}
The quantity is the quantum analogue of the classical multi-information between random variables \cite{han}. It is easy to verify that the product of the state marginals $\bigotimes_i\rho_{A_i}$ solves the minimization, ${\cal I}(\rho_{A_1,\ldots,A_n})=\sum_i S(\rho_{A_i})-S(\rho_{A_1,\ldots,A_n})$ \cite{modi}. Thus, quantitative bounds to the total system correlations in terms of purities are given by a straightforward generalization of Eq.~(\ref{thebound}).

In this letter, we have provided a strategy to evaluate coherence and
entanglement with limited laboratory resources. We have derived
bounds to the relative entropy of coherence and the coherent
information, which can be experimentally extracted by purity
measurements.  Although controlled swaps of large-dimensional systems
are hard to implement, in the case where the systems factorize into
qubits we can do the controlled swaps piece by piece~\cite{alves2004multipartite,jeong}.  We
verified the accuracy of our method by evaluating how tight our approximations are to
the actual entanglement/purity measures are.

The scheme is readily implementable in standard
quantum information testbeds, as optical lattices, ion traps and NMR
(Nuclear Magnetic Resonance) systems.  The scalability of the
measurement network makes purity detection employable in testing the
successful preparation of quantum superpositions in large
computational registers, certifying that a complex device has run a
truly quantum computation. The proposal could simplify the study of
key properties and structure of many-body complex systems, e.g. by
investigating phase transition of condensed matter through coherence
and entanglement detection.


\emph{Acknowledgments}---
We acknowledge T.~Peng for the insightful discussions. This work was supported by National Natural Science Foundation
of China, the 1000 Youth Fellowship program in China, the EPSRC (UK) through the grant EP/L01405X/1, the Wolfson College,
University of Oxford, the ARO under contract W911NF-14-1-012, and the NSF under Grant Number 11258444.

\bibliography{bibEntQuant}

\begin{thebibliography}{39}%
\makeatletter
\providecommand \@ifxundefined [1]{%
 \@ifx{#1\undefined}
}%
\providecommand \@ifnum [1]{%
 \ifnum #1\expandafter \@firstoftwo
 \else \expandafter \@secondoftwo
 \fi
}%
\providecommand \@ifx [1]{%
 \ifx #1\expandafter \@firstoftwo
 \else \expandafter \@secondoftwo
 \fi
}%
\providecommand \natexlab [1]{#1}%
\providecommand \enquote  [1]{``#1''}%
\providecommand \bibnamefont  [1]{#1}%
\providecommand \bibfnamefont [1]{#1}%
\providecommand \citenamefont [1]{#1}%
\providecommand \href@noop [0]{\@secondoftwo}%
\providecommand \href [0]{\begingroup \@sanitize@url \@href}%
\providecommand \@href[1]{\@@startlink{#1}\@@href}%
\providecommand \@@href[1]{\endgroup#1\@@endlink}%
\providecommand \@sanitize@url [0]{\catcode `\\12\catcode `\$12\catcode
  `\&12\catcode `\#12\catcode `\^12\catcode `\_12\catcode `\%12\relax}%
\providecommand \@@startlink[1]{}%
\providecommand \@@endlink[0]{}%
\providecommand \url  [0]{\begingroup\@sanitize@url \@url }%
\providecommand \@url [1]{\endgroup\@href {#1}{\urlprefix }}%
\providecommand \urlprefix  [0]{URL }%
\providecommand \Eprint [0]{\href }%
\providecommand \doibase [0]{http://dx.doi.org/}%
\providecommand \selectlanguage [0]{\@gobble}%
\providecommand \bibinfo  [0]{\@secondoftwo}%
\providecommand \bibfield  [0]{\@secondoftwo}%
\providecommand \translation [1]{[#1]}%
\providecommand \BibitemOpen [0]{}%
\providecommand \bibitemStop [0]{}%
\providecommand \bibitemNoStop [0]{.\EOS\space}%
\providecommand \EOS [0]{\spacefactor3000\relax}%
\providecommand \BibitemShut  [1]{\csname bibitem#1\endcsname}%
\let\auto@bib@innerbib\@empty
\bibitem [{\citenamefont {Horodecki}\ \emph {et~al.}(2009)\citenamefont
  {Horodecki}, \citenamefont {Horodecki}, \citenamefont {Horodecki},\ and\
  \citenamefont {Horodecki}}]{Horodecki09}%
  \BibitemOpen
  \bibfield  {author} {\bibinfo {author} {\bibfnamefont {R.}~\bibnamefont
  {Horodecki}}, \bibinfo {author} {\bibfnamefont {P.}~\bibnamefont
  {Horodecki}}, \bibinfo {author} {\bibfnamefont {M.}~\bibnamefont
  {Horodecki}}, \ and\ \bibinfo {author} {\bibfnamefont {K.}~\bibnamefont
  {Horodecki}},\ }\href {\doibase 10.1103/RevModPhys.81.865} {\bibfield
  {journal} {\bibinfo  {journal} {Rev. Mod. Phys.}\ }\textbf {\bibinfo {volume}
  {81}},\ \bibinfo {pages} {865} (\bibinfo {year} {2009})}\BibitemShut
  {NoStop}%
\bibitem [{\citenamefont {Bennett}\ and\ \citenamefont
  {Wiesner}(1992)}]{dense}%
  \BibitemOpen
  \bibfield  {author} {\bibinfo {author} {\bibfnamefont {C.~H.}\ \bibnamefont
  {Bennett}}\ and\ \bibinfo {author} {\bibfnamefont {S.~J.}\ \bibnamefont
  {Wiesner}},\ }\href@noop {} {\bibfield  {journal} {\bibinfo  {journal}
  {Physical review letters}\ }\textbf {\bibinfo {volume} {69}},\ \bibinfo
  {pages} {2881} (\bibinfo {year} {1992})}\BibitemShut {NoStop}%
\bibitem [{\citenamefont {Shor}(1997)}]{Shor97}%
  \BibitemOpen
  \bibfield  {author} {\bibinfo {author} {\bibfnamefont {P.~W.}\ \bibnamefont
  {Shor}},\ }\href {\doibase 10.1137/S0097539795293172} {\bibfield  {journal}
  {\bibinfo  {journal} {SIAM J. Comput.}\ }\textbf {\bibinfo {volume} {26}},\
  \bibinfo {pages} {1484} (\bibinfo {year} {1997})}\BibitemShut {NoStop}%
\bibitem [{\citenamefont {Ekert}(1991)}]{Ekert91}%
  \BibitemOpen
  \bibfield  {author} {\bibinfo {author} {\bibfnamefont {A.~K.}\ \bibnamefont
  {Ekert}},\ }\href {\doibase 10.1103/PhysRevLett.67.661} {\bibfield  {journal}
  {\bibinfo  {journal} {Phys. Rev. Lett.}\ }\textbf {\bibinfo {volume} {67}},\
  \bibinfo {pages} {661} (\bibinfo {year} {1991})}\BibitemShut {NoStop}%
\bibitem [{\citenamefont {Bennett}\ \emph {et~al.}(1993)\citenamefont
  {Bennett}, \citenamefont {Brassard}, \citenamefont {Cr{\'e}peau},
  \citenamefont {Jozsa}, \citenamefont {Peres},\ and\ \citenamefont
  {Wootters}}]{telep}%
  \BibitemOpen
  \bibfield  {author} {\bibinfo {author} {\bibfnamefont {C.~H.}\ \bibnamefont
  {Bennett}}, \bibinfo {author} {\bibfnamefont {G.}~\bibnamefont {Brassard}},
  \bibinfo {author} {\bibfnamefont {C.}~\bibnamefont {Cr{\'e}peau}}, \bibinfo
  {author} {\bibfnamefont {R.}~\bibnamefont {Jozsa}}, \bibinfo {author}
  {\bibfnamefont {A.}~\bibnamefont {Peres}}, \ and\ \bibinfo {author}
  {\bibfnamefont {W.~K.}\ \bibnamefont {Wootters}},\ }\href@noop {} {\bibfield
  {journal} {\bibinfo  {journal} {Physical review letters}\ }\textbf {\bibinfo
  {volume} {70}},\ \bibinfo {pages} {1895} (\bibinfo {year}
  {1993})}\BibitemShut {NoStop}%
\bibitem [{\citenamefont {Amico}\ \emph {et~al.}(2008)\citenamefont {Amico},
  \citenamefont {Fazio}, \citenamefont {Osterloh},\ and\ \citenamefont
  {Vedral}}]{vedral}%
  \BibitemOpen
  \bibfield  {author} {\bibinfo {author} {\bibfnamefont {L.}~\bibnamefont
  {Amico}}, \bibinfo {author} {\bibfnamefont {R.}~\bibnamefont {Fazio}},
  \bibinfo {author} {\bibfnamefont {A.}~\bibnamefont {Osterloh}}, \ and\
  \bibinfo {author} {\bibfnamefont {V.}~\bibnamefont {Vedral}},\ }\href@noop {}
  {\bibfield  {journal} {\bibinfo  {journal} {Reviews of Modern Physics}\
  }\textbf {\bibinfo {volume} {80}},\ \bibinfo {pages} {517} (\bibinfo {year}
  {2008})}\BibitemShut {NoStop}%
\bibitem [{\citenamefont {Giovannetti}\ \emph {et~al.}(2011)\citenamefont
  {Giovannetti}, \citenamefont {Lloyd},\ and\ \citenamefont
  {Maccone}}]{metrorev}%
  \BibitemOpen
  \bibfield  {author} {\bibinfo {author} {\bibfnamefont {V.}~\bibnamefont
  {Giovannetti}}, \bibinfo {author} {\bibfnamefont {S.}~\bibnamefont {Lloyd}},
  \ and\ \bibinfo {author} {\bibfnamefont {L.}~\bibnamefont {Maccone}},\
  }\href@noop {} {\bibfield  {journal} {\bibinfo  {journal} {Nature photonics}\
  }\textbf {\bibinfo {volume} {5}},\ \bibinfo {pages} {222} (\bibinfo {year}
  {2011})}\BibitemShut {NoStop}%
\bibitem [{Note1()}]{Note1}%
  \BibitemOpen
  \bibinfo {note} {The von Neumann entropy of the reduced density matrix of
  either subsystem is an entanglement measure for bipartite pure
  states}\BibitemShut {NoStop}%
\bibitem [{\citenamefont {Plenio}\ and\ \citenamefont
  {Virmani}(2007)}]{Measures}%
  \BibitemOpen
  \bibfield  {author} {\bibinfo {author} {\bibfnamefont {M.~B.}\ \bibnamefont
  {Plenio}}\ and\ \bibinfo {author} {\bibfnamefont {S.}~\bibnamefont
  {Virmani}},\ }\href {http://dl.acm.org/citation.cfm?id=2011706.2011707}
  {\bibfield  {journal} {\bibinfo  {journal} {Quantum Info. Comput.}\ }\textbf
  {\bibinfo {volume} {7}},\ \bibinfo {pages} {1} (\bibinfo {year}
  {2007})}\BibitemShut {NoStop}%
\bibitem [{\citenamefont {Bennett}\ \emph {et~al.}(1996)\citenamefont
  {Bennett}, \citenamefont {DiVincenzo}, \citenamefont {Smolin},\ and\
  \citenamefont {Wootters}}]{BDSW}%
  \BibitemOpen
  \bibfield  {author} {\bibinfo {author} {\bibfnamefont {C.~H.}\ \bibnamefont
  {Bennett}}, \bibinfo {author} {\bibfnamefont {D.~P.}\ \bibnamefont
  {DiVincenzo}}, \bibinfo {author} {\bibfnamefont {J.~A.}\ \bibnamefont
  {Smolin}}, \ and\ \bibinfo {author} {\bibfnamefont {W.~K.}\ \bibnamefont
  {Wootters}},\ }\href {\doibase 10.1103/PhysRevA.54.3824} {\bibfield
  {journal} {\bibinfo  {journal} {Phys. Rev. A}\ }\textbf {\bibinfo {volume}
  {54}},\ \bibinfo {pages} {3824} (\bibinfo {year} {1996})}\BibitemShut
  {NoStop}%
\bibitem [{\citenamefont {Hayden}\ \emph {et~al.}(2001)\citenamefont {Hayden},
  \citenamefont {Horodecki},\ and\ \citenamefont
  {Terhal}}]{hayden2001asymptotic}%
  \BibitemOpen
  \bibfield  {author} {\bibinfo {author} {\bibfnamefont {P.~M.}\ \bibnamefont
  {Hayden}}, \bibinfo {author} {\bibfnamefont {M.}~\bibnamefont {Horodecki}}, \
  and\ \bibinfo {author} {\bibfnamefont {B.~M.}\ \bibnamefont {Terhal}},\
  }\href@noop {} {\bibfield  {journal} {\bibinfo  {journal} {Journal of Physics
  A: Mathematical and General}\ }\textbf {\bibinfo {volume} {34}},\ \bibinfo
  {pages} {6891} (\bibinfo {year} {2001})}\BibitemShut {NoStop}%
\bibitem [{\citenamefont {Wootters}(1998)}]{wootters}%
  \BibitemOpen
  \bibfield  {author} {\bibinfo {author} {\bibfnamefont {W.~K.}\ \bibnamefont
  {Wootters}},\ }\href {\doibase 10.1103/PhysRevLett.80.2245} {\bibfield
  {journal} {\bibinfo  {journal} {Phys. Rev. Lett.}\ }\textbf {\bibinfo
  {volume} {80}},\ \bibinfo {pages} {2245} (\bibinfo {year}
  {1998})}\BibitemShut {NoStop}%
\bibitem [{\citenamefont {Vidal}\ and\ \citenamefont
  {Werner}(2002)}]{vidal2002computable}%
  \BibitemOpen
  \bibfield  {author} {\bibinfo {author} {\bibfnamefont {G.}~\bibnamefont
  {Vidal}}\ and\ \bibinfo {author} {\bibfnamefont {R.~F.}\ \bibnamefont
  {Werner}},\ }\href@noop {} {\bibfield  {journal} {\bibinfo  {journal}
  {Physical Review A}\ }\textbf {\bibinfo {volume} {65}},\ \bibinfo {pages}
  {032314} (\bibinfo {year} {2002})}\BibitemShut {NoStop}%
\bibitem [{\citenamefont {G{\"u}hne}\ and\ \citenamefont
  {T{\'o}th}(2009)}]{entdet}%
  \BibitemOpen
  \bibfield  {author} {\bibinfo {author} {\bibfnamefont {O.}~\bibnamefont
  {G{\"u}hne}}\ and\ \bibinfo {author} {\bibfnamefont {G.}~\bibnamefont
  {T{\'o}th}},\ }\href@noop {} {\bibfield  {journal} {\bibinfo  {journal}
  {Physics Reports}\ }\textbf {\bibinfo {volume} {474}},\ \bibinfo {pages} {1}
  (\bibinfo {year} {2009})}\BibitemShut {NoStop}%
\bibitem [{\citenamefont {Terhal}(2000)}]{terhal}%
  \BibitemOpen
  \bibfield  {author} {\bibinfo {author} {\bibfnamefont {B.~M.}\ \bibnamefont
  {Terhal}},\ }\href@noop {} {\bibfield  {journal} {\bibinfo  {journal}
  {Physics Letters A}\ }\textbf {\bibinfo {volume} {271}},\ \bibinfo {pages}
  {319} (\bibinfo {year} {2000})}\BibitemShut {NoStop}%
\bibitem [{\citenamefont {Strobel}\ \emph {et~al.}(2014)\citenamefont
  {Strobel}, \citenamefont {Muessel}, \citenamefont {Linnemann}, \citenamefont
  {Zibold}, \citenamefont {Hume}, \citenamefont {Pezz{\`e}}, \citenamefont
  {Smerzi},\ and\ \citenamefont {Oberthaler}}]{oberthaler}%
  \BibitemOpen
  \bibfield  {author} {\bibinfo {author} {\bibfnamefont {H.}~\bibnamefont
  {Strobel}}, \bibinfo {author} {\bibfnamefont {W.}~\bibnamefont {Muessel}},
  \bibinfo {author} {\bibfnamefont {D.}~\bibnamefont {Linnemann}}, \bibinfo
  {author} {\bibfnamefont {T.}~\bibnamefont {Zibold}}, \bibinfo {author}
  {\bibfnamefont {D.~B.}\ \bibnamefont {Hume}}, \bibinfo {author}
  {\bibfnamefont {L.}~\bibnamefont {Pezz{\`e}}}, \bibinfo {author}
  {\bibfnamefont {A.}~\bibnamefont {Smerzi}}, \ and\ \bibinfo {author}
  {\bibfnamefont {M.~K.}\ \bibnamefont {Oberthaler}},\ }\href@noop {}
  {\bibfield  {journal} {\bibinfo  {journal} {Science}\ }\textbf {\bibinfo
  {volume} {345}},\ \bibinfo {pages} {424} (\bibinfo {year}
  {2014})}\BibitemShut {NoStop}%
\bibitem [{\citenamefont {Hauke}\ \emph {et~al.}(2016)\citenamefont {Hauke},
  \citenamefont {Heyl}, \citenamefont {Tagliacozzo},\ and\ \citenamefont
  {Zoller}}]{zoller}%
  \BibitemOpen
  \bibfield  {author} {\bibinfo {author} {\bibfnamefont {P.}~\bibnamefont
  {Hauke}}, \bibinfo {author} {\bibfnamefont {M.}~\bibnamefont {Heyl}},
  \bibinfo {author} {\bibfnamefont {L.}~\bibnamefont {Tagliacozzo}}, \ and\
  \bibinfo {author} {\bibfnamefont {P.}~\bibnamefont {Zoller}},\ }\href@noop {}
  {\bibfield  {journal} {\bibinfo  {journal} {Nature Physics}\ } (\bibinfo
  {year} {2016})}\BibitemShut {NoStop}%
\bibitem [{\citenamefont {Zhang}\ \emph {et~al.}(2008)\citenamefont {Zhang},
  \citenamefont {Gong}, \citenamefont {Zhang},\ and\ \citenamefont
  {Guo}}]{hefei}%
  \BibitemOpen
  \bibfield  {author} {\bibinfo {author} {\bibfnamefont {C.-J.}\ \bibnamefont
  {Zhang}}, \bibinfo {author} {\bibfnamefont {Y.-X.}\ \bibnamefont {Gong}},
  \bibinfo {author} {\bibfnamefont {Y.-S.}\ \bibnamefont {Zhang}}, \ and\
  \bibinfo {author} {\bibfnamefont {G.-C.}\ \bibnamefont {Guo}},\ }\href
  {\doibase 10.1103/PhysRevA.78.042308} {\bibfield  {journal} {\bibinfo
  {journal} {Phys. Rev. A}\ }\textbf {\bibinfo {volume} {78}},\ \bibinfo
  {pages} {042308} (\bibinfo {year} {2008})}\BibitemShut {NoStop}%
\bibitem [{\citenamefont {Alves}\ and\ \citenamefont
  {Jaksch}(2004)}]{alves2004multipartite}%
  \BibitemOpen
  \bibfield  {author} {\bibinfo {author} {\bibfnamefont {C.~M.}\ \bibnamefont
  {Alves}}\ and\ \bibinfo {author} {\bibfnamefont {D.}~\bibnamefont {Jaksch}},\
  }\href@noop {} {\bibfield  {journal} {\bibinfo  {journal} {Physical review
  letters}\ }\textbf {\bibinfo {volume} {93}},\ \bibinfo {pages} {110501}
  (\bibinfo {year} {2004})}\BibitemShut {NoStop}%
\bibitem [{\citenamefont {Islam}\ \emph {et~al.}(2015)\citenamefont {Islam},
  \citenamefont {Ma}, \citenamefont {Preiss}, \citenamefont {Tai},
  \citenamefont {Lukin}, \citenamefont {Rispoli},\ and\ \citenamefont
  {Greiner}}]{Experiment}%
  \BibitemOpen
  \bibfield  {author} {\bibinfo {author} {\bibfnamefont {R.}~\bibnamefont
  {Islam}}, \bibinfo {author} {\bibfnamefont {R.}~\bibnamefont {Ma}}, \bibinfo
  {author} {\bibfnamefont {P.~M.}\ \bibnamefont {Preiss}}, \bibinfo {author}
  {\bibfnamefont {M.~E.}\ \bibnamefont {Tai}}, \bibinfo {author} {\bibfnamefont
  {A.}~\bibnamefont {Lukin}}, \bibinfo {author} {\bibfnamefont
  {M.}~\bibnamefont {Rispoli}}, \ and\ \bibinfo {author} {\bibfnamefont
  {M.}~\bibnamefont {Greiner}},\ }\href {\doibase doi:10.1038/nature15750}
  {\bibfield  {journal} {\bibinfo  {journal} {Nature}\ }\textbf {\bibinfo
  {volume} {528}},\ \bibinfo {pages} {77} (\bibinfo {year} {2015})}\BibitemShut
  {NoStop}%
\bibitem [{\citenamefont {{\.Z}yczkowski}(2003)}]{zyczkowski2003renyi}%
  \BibitemOpen
  \bibfield  {author} {\bibinfo {author} {\bibfnamefont {K.}~\bibnamefont
  {{\.Z}yczkowski}},\ }\href@noop {} {\bibfield  {journal} {\bibinfo  {journal}
  {Open Systems \& Information Dynamics}\ }\textbf {\bibinfo {volume} {10}},\
  \bibinfo {pages} {297} (\bibinfo {year} {2003})}\BibitemShut {NoStop}%
\bibitem [{\citenamefont {Horodecki}\ \emph {et~al.}(2000)\citenamefont
  {Horodecki}, \citenamefont {Horodecki},\ and\ \citenamefont
  {Horodecki}}]{horodecki2000limits}%
  \BibitemOpen
  \bibfield  {author} {\bibinfo {author} {\bibfnamefont {M.}~\bibnamefont
  {Horodecki}}, \bibinfo {author} {\bibfnamefont {P.}~\bibnamefont
  {Horodecki}}, \ and\ \bibinfo {author} {\bibfnamefont {R.}~\bibnamefont
  {Horodecki}},\ }\href@noop {} {\bibfield  {journal} {\bibinfo  {journal}
  {Physical Review Letters}\ }\textbf {\bibinfo {volume} {84}},\ \bibinfo
  {pages} {2014} (\bibinfo {year} {2000})}\BibitemShut {NoStop}%
\bibitem [{\citenamefont {Schumacher}\ and\ \citenamefont
  {Nielsen}(1996)}]{schum}%
  \BibitemOpen
  \bibfield  {author} {\bibinfo {author} {\bibfnamefont {B.}~\bibnamefont
  {Schumacher}}\ and\ \bibinfo {author} {\bibfnamefont {M.~A.}\ \bibnamefont
  {Nielsen}},\ }\href@noop {} {\bibfield  {journal} {\bibinfo  {journal}
  {Physical Review A}\ }\textbf {\bibinfo {volume} {54}},\ \bibinfo {pages}
  {2629} (\bibinfo {year} {1996})}\BibitemShut {NoStop}%
\bibitem [{\citenamefont {Lloyd}(1997)}]{Lloyd}%
  \BibitemOpen
  \bibfield  {author} {\bibinfo {author} {\bibfnamefont {S.}~\bibnamefont
  {Lloyd}},\ }\href {\doibase 10.1103/PhysRevA.55.1613} {\bibfield  {journal}
  {\bibinfo  {journal} {Phys. Rev. A}\ }\textbf {\bibinfo {volume} {55}},\
  \bibinfo {pages} {1613} (\bibinfo {year} {1997})}\BibitemShut {NoStop}%
\bibitem [{\citenamefont {Devetak}(2005)}]{Devetak}%
  \BibitemOpen
  \bibfield  {author} {\bibinfo {author} {\bibfnamefont {I.}~\bibnamefont
  {Devetak}},\ }\href@noop {} {\bibfield  {journal} {\bibinfo  {journal} {IEEE
  Transactions on Information Theory}\ }\textbf {\bibinfo {volume} {51}},\
  \bibinfo {pages} {44} (\bibinfo {year} {2005})}\BibitemShut {NoStop}%
\bibitem [{\citenamefont {Shor}(2002)}]{Shor}%
  \BibitemOpen
  \bibfield  {author} {\bibinfo {author} {\bibfnamefont {P.~W.}\ \bibnamefont
  {Shor}},\ }in\ \href@noop {} {\emph {\bibinfo {booktitle} {lecture notes,
  MSRI Workshop on Quantum Computation}}}\ (\bibinfo {year} {2002})\BibitemShut
  {NoStop}%
\bibitem [{\citenamefont {Devetak}\ and\ \citenamefont
  {Winter}(2005)}]{DevetakWinter}%
  \BibitemOpen
  \bibfield  {author} {\bibinfo {author} {\bibfnamefont {I.}~\bibnamefont
  {Devetak}}\ and\ \bibinfo {author} {\bibfnamefont {A.}~\bibnamefont
  {Winter}},\ }in\ \href@noop {} {\emph {\bibinfo {booktitle} {Proceedings of
  the Royal Society of London A: Mathematical, Physical and Engineering
  Sciences}}},\ Vol.\ \bibinfo {volume} {461}\ (\bibinfo {organization} {The
  Royal Society},\ \bibinfo {year} {2005})\ pp.\ \bibinfo {pages}
  {207--235}\BibitemShut {NoStop}%
\bibitem [{\citenamefont {Horodecki}\ \emph {et~al.}(2005)\citenamefont
  {Horodecki}, \citenamefont {Oppenheim},\ and\ \citenamefont
  {Winter}}]{horodecki2005partial}%
  \BibitemOpen
  \bibfield  {author} {\bibinfo {author} {\bibfnamefont {M.}~\bibnamefont
  {Horodecki}}, \bibinfo {author} {\bibfnamefont {J.}~\bibnamefont
  {Oppenheim}}, \ and\ \bibinfo {author} {\bibfnamefont {A.}~\bibnamefont
  {Winter}},\ }\href@noop {} {\bibfield  {journal} {\bibinfo  {journal}
  {Nature}\ }\textbf {\bibinfo {volume} {436}},\ \bibinfo {pages} {673}
  (\bibinfo {year} {2005})}\BibitemShut {NoStop}%
\bibitem [{\citenamefont {Streltsov}\ \emph {et~al.}(2016)\citenamefont
  {Streltsov}, \citenamefont {Adesso},\ and\ \citenamefont {Plenio}}]{cohrev}%
  \BibitemOpen
  \bibfield  {author} {\bibinfo {author} {\bibfnamefont {A.}~\bibnamefont
  {Streltsov}}, \bibinfo {author} {\bibfnamefont {G.}~\bibnamefont {Adesso}}, \
  and\ \bibinfo {author} {\bibfnamefont {M.~B.}\ \bibnamefont {Plenio}},\
  }\href@noop {} {\bibfield  {journal} {\bibinfo  {journal} {arXiv preprint
  arXiv:1609.02439}\ } (\bibinfo {year} {2016})}\BibitemShut {NoStop}%
\bibitem [{\citenamefont {Baumgratz}\ \emph {et~al.}(2014)\citenamefont
  {Baumgratz}, \citenamefont {Cramer},\ and\ \citenamefont {Plenio}}]{BCP}%
  \BibitemOpen
  \bibfield  {author} {\bibinfo {author} {\bibfnamefont {T.}~\bibnamefont
  {Baumgratz}}, \bibinfo {author} {\bibfnamefont {M.}~\bibnamefont {Cramer}}, \
  and\ \bibinfo {author} {\bibfnamefont {M.}~\bibnamefont {Plenio}},\
  }\href@noop {} {\bibfield  {journal} {\bibinfo  {journal} {Physical review
  letters}\ }\textbf {\bibinfo {volume} {113}},\ \bibinfo {pages} {140401}
  (\bibinfo {year} {2014})}\BibitemShut {NoStop}%
\bibitem [{\citenamefont {Herbut}(2005)}]{herbut}%
  \BibitemOpen
  \bibfield  {author} {\bibinfo {author} {\bibfnamefont {F.}~\bibnamefont
  {Herbut}},\ }\href@noop {} {\bibfield  {journal} {\bibinfo  {journal}
  {Journal of Physics A: Mathematical and General}\ }\textbf {\bibinfo {volume}
  {38}},\ \bibinfo {pages} {2959} (\bibinfo {year} {2005})}\BibitemShut
  {NoStop}%
\bibitem [{\citenamefont {Winter}\ and\ \citenamefont {Yang}(2016)}]{yang}%
  \BibitemOpen
  \bibfield  {author} {\bibinfo {author} {\bibfnamefont {A.}~\bibnamefont
  {Winter}}\ and\ \bibinfo {author} {\bibfnamefont {D.}~\bibnamefont {Yang}},\
  }\href@noop {} {\bibfield  {journal} {\bibinfo  {journal} {Physical review
  letters}\ }\textbf {\bibinfo {volume} {116}},\ \bibinfo {pages} {120404}
  (\bibinfo {year} {2016})}\BibitemShut {NoStop}%
\bibitem [{\citenamefont {Brun}(2004)}]{brun}%
  \BibitemOpen
  \bibfield  {author} {\bibinfo {author} {\bibfnamefont {T.~A.}\ \bibnamefont
  {Brun}},\ }\href@noop {} {\bibfield  {journal} {\bibinfo  {journal} {arXiv
  preprint quant-ph/0401067}\ } (\bibinfo {year} {2004})}\BibitemShut {NoStop}%
\bibitem [{\citenamefont {Ekert}\ \emph {et~al.}(2002)\citenamefont {Ekert},
  \citenamefont {Alves}, \citenamefont {Oi}, \citenamefont {Horodecki},
  \citenamefont {Horodecki},\ and\ \citenamefont {Kwek}}]{ekert}%
  \BibitemOpen
  \bibfield  {author} {\bibinfo {author} {\bibfnamefont {A.~K.}\ \bibnamefont
  {Ekert}}, \bibinfo {author} {\bibfnamefont {C.~M.}\ \bibnamefont {Alves}},
  \bibinfo {author} {\bibfnamefont {D.~K.}\ \bibnamefont {Oi}}, \bibinfo
  {author} {\bibfnamefont {M.}~\bibnamefont {Horodecki}}, \bibinfo {author}
  {\bibfnamefont {P.}~\bibnamefont {Horodecki}}, \ and\ \bibinfo {author}
  {\bibfnamefont {L.~C.}\ \bibnamefont {Kwek}},\ }\href@noop {} {\bibfield
  {journal} {\bibinfo  {journal} {Physical review letters}\ }\textbf {\bibinfo
  {volume} {88}},\ \bibinfo {pages} {217901} (\bibinfo {year}
  {2002})}\BibitemShut {NoStop}%
\bibitem [{\citenamefont {Filip}(2002)}]{filip}%
  \BibitemOpen
  \bibfield  {author} {\bibinfo {author} {\bibfnamefont {R.}~\bibnamefont
  {Filip}},\ }\href@noop {} {\bibfield  {journal} {\bibinfo  {journal}
  {Physical Review A}\ }\textbf {\bibinfo {volume} {65}},\ \bibinfo {pages}
  {062320} (\bibinfo {year} {2002})}\BibitemShut {NoStop}%
\bibitem [{\citenamefont {Girolami}(2014)}]{Girolami14}%
  \BibitemOpen
  \bibfield  {author} {\bibinfo {author} {\bibfnamefont {D.}~\bibnamefont
  {Girolami}},\ }\href {\doibase 10.1103/PhysRevLett.113.170401} {\bibfield
  {journal} {\bibinfo  {journal} {Phys. Rev. Lett.}\ }\textbf {\bibinfo
  {volume} {113}},\ \bibinfo {pages} {170401} (\bibinfo {year}
  {2014})}\BibitemShut {NoStop}%
\bibitem [{\citenamefont {Han}(1978)}]{han}%
  \BibitemOpen
  \bibfield  {author} {\bibinfo {author} {\bibfnamefont {T.~S.}\ \bibnamefont
  {Han}},\ }\href@noop {} {\bibfield  {journal} {\bibinfo  {journal}
  {Information and Control}\ }\textbf {\bibinfo {volume} {36}},\ \bibinfo
  {pages} {133} (\bibinfo {year} {1978})}\BibitemShut {NoStop}%
\bibitem [{\citenamefont {Modi}\ \emph {et~al.}(2010)\citenamefont {Modi},
  \citenamefont {Paterek}, \citenamefont {Son}, \citenamefont {Vedral},\ and\
  \citenamefont {Williamson}}]{modi}%
  \BibitemOpen
  \bibfield  {author} {\bibinfo {author} {\bibfnamefont {K.}~\bibnamefont
  {Modi}}, \bibinfo {author} {\bibfnamefont {T.}~\bibnamefont {Paterek}},
  \bibinfo {author} {\bibfnamefont {W.}~\bibnamefont {Son}}, \bibinfo {author}
  {\bibfnamefont {V.}~\bibnamefont {Vedral}}, \ and\ \bibinfo {author}
  {\bibfnamefont {M.}~\bibnamefont {Williamson}},\ }\href@noop {} {\bibfield
  {journal} {\bibinfo  {journal} {Physical review letters}\ }\textbf {\bibinfo
  {volume} {104}},\ \bibinfo {pages} {080501} (\bibinfo {year}
  {2010})}\BibitemShut {NoStop}%
\bibitem [{\citenamefont {Jeong}\ \emph {et~al.}(2014)\citenamefont {Jeong},
  \citenamefont {Noh}, \citenamefont {Bae}, \citenamefont {Angelakis},\ and\
  \citenamefont {Ralph}}]{jeong}%
  \BibitemOpen
  \bibfield  {author} {\bibinfo {author} {\bibfnamefont {H.}~\bibnamefont
  {Jeong}}, \bibinfo {author} {\bibfnamefont {C.}~\bibnamefont {Noh}}, \bibinfo
  {author} {\bibfnamefont {S.}~\bibnamefont {Bae}}, \bibinfo {author}
  {\bibfnamefont {D.}~\bibnamefont {Angelakis}}, \ and\ \bibinfo {author}
  {\bibfnamefont {T.}~\bibnamefont {Ralph}},\ }\href@noop {} {\bibfield
  {journal} {\bibinfo  {journal} {J. of Opt. Soc. of Am. B}\ }\textbf {\bibinfo
  {volume} {31}},\ \bibinfo {pages} {3057} (\bibinfo {year}
  {2014})}\BibitemShut {NoStop}%
\end{thebibliography}%
\widetext
\clearpage
\appendix

\section*{Appendix: Derivation of the bounds to relative entropy of coherence and coherent information, Eqs. 4,6 of the main text}

Given a quantum state $\rho$ in a $d$-dimensional Hilbert space, our task is to bound the Von Neumann entropy of $\rho$ with a function of the state purity $\gamma(\rho):=\mathrm{Tr}(\rho^2).$
The spectral decomposition of the quantum state is $\rho = \sum_{i=1}^dx_i\ket{\psi_i}\bra{\psi_i}$, where $\{\ket{\psi_i}\}$ forms an orthonormal basis of the $d$-dimensional Hilbert space. The variational problem is then formulated as
\begin{equation}\label{Eq:maxd}
\begin{aligned}
\max/\min S(\rho)&= -\sum_{i=1}^dx_i\log(x_i)\\
s.t.~~ &\sum_{i=1}^dx_i^2 = \gamma\\
&\sum_{i=1}^dx_i=1\\
&0\le x_i \le1, \forall i,
\end{aligned}
\end{equation}
where $\gamma=\mathrm{Tr} \rho^2$ is the purity of $\rho$.

Intuitively, the the vector $x$ that maximize $S$ is the one that spread as uniformly as possible; while the vector $x$ that minimize $S$ is the one that has the minimal number of nonzero large values. In the following, we will analytically solve this problem and confirm this intuition.

\subsection{Maximization}
First, we focus on the maximization problem with $d=3$. Note that when $d=2$, the solution to the constraints of Eq.~\eqref{Eq:maxd} is unique and the optimization problem will be trivial. Without loss of generality, we assume $x_1\ge x_2\ge x_3$. Then the problem can be stated as
\begin{equation}\label{Eq:max3}
\begin{aligned}
\max S(\rho)&= -x_1\log(x_1)-x_2\log(x_2)-x_3\log(x_3)\\
s.t.~~ &x_1^2+x_2^2+x_3^2 = \gamma\\
&x_1+x_2+x_3=1\\
&1\ge x_1\ge x_2\ge x_3\ge0. \\
\end{aligned}
\end{equation}
We prove that the maximum is reached with the following Lemma.
\begin{lemma}\label{lamme:max3}
The solution to the maximization problem in Eq.~\eqref{Eq:max3} is given by
\begin{equation}\label{lamme:max3solution}
\begin{aligned}
  x_1 &= \frac{1}{3}+\sqrt{\frac{2}{3}\left(\gamma-\frac{1}{3}\right)},\nonumber\\
  x_2 & = x_3 = \frac{1-x_1}{2}.
\end{aligned}
\end{equation}
\end{lemma}

\begin{proof}
The differential of the entropy function $S(\rho)$ and the constraints are given by
\begin{equation}\label{Eq:}
\begin{aligned}
 dS&= -(1+\log x_1)dx_1-(1+\log x_2)dx_2-(1+\log x_3)dx_3\nonumber\\
 \end{aligned}
\end{equation}
and
\begin{equation}\label{Eq:ConstraintDiff}
\begin{aligned}
 &x_1dx_1+x_2dx_2+x_3dx_3 = 0\nonumber\\
 &dx_1+dx_2+dx_3=0,
  \end{aligned}
\end{equation}
respectively. We rewrite Eq.~\eqref{Eq:ConstraintDiff} to
\begin{equation}\label{Eq:}
\begin{aligned}
 dx_1=&-\frac{(x_3-x_2)}{x_1-x_2}dx_3,\nonumber\\
 dx_2=&-\frac{(x_1-x_3)}{x_1-x_2}dx_3.
  \end{aligned}
\end{equation}
Thus, the differential of the entropy function becomes
\begin{equation}
\begin{aligned}
dS(\rho)&=\frac{dx_3}{x_1-x_2}[(x_3-x_2)\log x_1 +(x_1-x_3)\log x_2 +(x_2-x_1)\log x_3]\nonumber\\
&= (x_2-x_3)\left[-\frac{\log x_1-\log x_2}{x_1-x_2}+\frac{\log x_3-\log x_2}{x_3-x_2}\right]dx_3
\end{aligned}
\end{equation}
Since the function $\log x$ is convex for $x\in[0,1]$,
for $x_1\ge x_2\ge x_3$,
\begin{equation}
\frac{\log x_1-\log x_2}{x_1-x_2} \le \frac{\log x_3-\log x_2}{x_3-x_2}.\nonumber
\end{equation}
Thus, $dS(\rho)/dx_3\ge 0$. To reach the maximum of $S(\rho)$, we thus only need to set $x_3$ to be its maximum, which happens when $x_2=x_3$. Together with the constraints, then we can solve the equations and show that the solution to the maximization problem is given in Eq.~\eqref{lamme:max3solution}.
\end{proof}

Now, we can solve the maximization problem of Eq.~\eqref{Eq:maxd} for a general case of $d$.
\begin{theorem}\label{theorem:max}
Suppose $x_1\ge x_2 \ge \dots x_d$, the solution to the maximization problem in Eq.~\eqref{Eq:maxd} is
\begin{equation}\label{Eq:solutionmaxd}
\begin{aligned}
  x_1 &= \frac{1}{d}+\sqrt{\frac{d-1}{d}\left(\gamma-\frac{1}{d}\right)},\\
  x_2 &=x_3=\dots=x_d = \frac{1-x_1}{d-1}.
\end{aligned}
\end{equation}
\end{theorem}
\begin{proof}
The solution in Eq.~\eqref{Eq:solutionmaxd} is exactly determined when setting $x_2=x_3=\dots=x_d$. Suppose the maximization problem solution is not this one, then we must have that $x_2>x_d$. In the following, we prove the contradiction by showing that changing the values of $x_1, x_2, x_d$ would make the entropy $S(\rho)$ larger, while fixing all other values ($x_3, x_4, \dots, x_{d-1}$) and the constraints. Now the constraints for $x_1$, $x_2$, and $x_d$ becomes
\begin{equation}\label{Eq:cons12d}
\begin{aligned}
&x_1^2+x_2^2+x_d^2 = a\nonumber\\
 &x_1+x_2+x_d=b.
 \end{aligned}
\end{equation}
By defining $x_1' = x_1/b$, $x_2' = x_2/b$, $x_d' = x_d/b$, the relations become
\begin{equation}\label{Eq:cons12dp}
\begin{aligned}
x_1'^2+x_2'^2+x_d'^2 &= a/b^2\nonumber\\
 x_1'+x_2'+x_d'&=1.
 \end{aligned}
\end{equation}
The entropy function is
\begin{equation}\label{Eq:}
\begin{aligned}
S(\rho)&=-\sum_{i=1}^d x_i\log(x_i),\nonumber\\
 &=S_{1,2,d}(\rho) +S_r(\rho),
 \end{aligned}
\end{equation}
where $S_{1,2,d}(\rho) = -x_1\log(x_1)-x_2\log(x_2)-x_d\log(x_d)$ and $S_r(\rho)=-\sum_{i=3}^{d-1} x_i\log(x_i)$. Since $S_r(\rho)$ is fixed, we need to maximize $S_{1,2,d}(\rho)$, which can also be represented as
\begin{equation}\label{Eq:}
\begin{aligned}
S_{1,2,d}(\rho) &= -bx_1'\log(bx_1')-bx_2'\log(bx_2')-bx_d'\log(bx_d')\nonumber\\
&= b[-x_1'\log(x_1')-x_2'\log(x_2')-x_d'\log(x_d')]-b\log b\\
 \end{aligned}
\end{equation}
Denoting $S_{1,2,d}'(\rho) = -x_1'\log(x_1')-x_2'\log(x_2')-x_d'\log(x_d')$, this optimization problem has the same form of Eq.~\eqref{Eq:max3}. Then Lemma \ref{lamme:max3} indicates that the maximum of $S_{1,2,d}'(\rho)$ given the constraints in Eq.~\eqref{Eq:cons12dp} is reached when $x_2'=x_d'$. In other words, the maximum of $S_{1,2,d}(\rho)$ given the constrains of Eq.~\eqref{Eq:cons12d} is saturated with $x_2=x_d$, which contradicts with $x_2>x_d$. Therefore, the solution to the maximization problem is given by Eq.~\eqref{Eq:solutionmaxd}.
\end{proof}

\subsection{Minimization}
Now, we consider the solution to the minimization of Eq.~\eqref{Eq:maxd}.  Similarly, we first consider the minimization with $d=3$ and $x_1\ge x_2\ge x_3$,
\begin{equation}\label{Eq:min3}
\begin{aligned}
\min S(\rho)&= -x_1\log(x_1)-x_2\log(x_2)-x_3\log(x_3)\\
s.t.~~ &x_1^2+x_2^2+x_3^2 = \gamma\\
 &x_1+x_2+x_3=1\\
&1\ge x_1\ge x_2\ge x_3\ge 0. \\
 \end{aligned}
\end{equation}
\begin{lemma}\label{lamme:min3}
The solution to the minimization problem in Eq.~\eqref{Eq:min3} is reached either when $x_1=x_2$ or $x_3=0$.
\end{lemma}
\begin{proof}
From the proof of Lemma \ref{lamme:max3}, we already showed that $dS(\rho)/dx_3\ge 0$. Therefore, the lower bound of $S(\rho)$ is reached when $x_3$ takes its minimum. As $2(x_1^2+x_2^2)\ge (x_1+x_2)^2$, according to Eq.~\eqref{Eq:min3}, we have
\begin{equation}\label{Eq:}
\begin{aligned}
2(\gamma-x_3^2)\ge (1-x_3)^2.\nonumber
 \end{aligned}
\end{equation}
The lower bound for $x_3$ is
\begin{equation}
\begin{aligned}
x_3\ge \max \left\{0, \frac{1-\sqrt{6\gamma - 2}}{3}\right\}\nonumber
 \end{aligned}
\end{equation}
Thus, when $\gamma \ge 1/2$, the minimal possible value for $x_3$ is 0. When $1/3\le \gamma < 1/2$, the minimal possible value for $x_3$ is $\frac{1-\sqrt{6\gamma - 2}}{3}$ and $x_1=x_2=(1-x_3)/2$. Note that $\gamma\ge1/3$ for $d=3$.
\end{proof}

Now, we can show the general solution to the minimization of Eq.~\eqref{Eq:maxd}.
\begin{theorem}\label{theorem:min}
Suppose $x_1\ge x_2 \ge \dots x_k$, the solution to the minimization problem in Eq.~\eqref{Eq:maxd} is
\begin{equation}\label{Eq:solutionmind}
\begin{aligned}
x_1=x_2=\dots=x_{k-1}&=\frac{1-\alpha}{k-1},\\
x_{k} &= \alpha,\\
x_{k+1}=\dots=x_d &= 0.
\end{aligned}
\end{equation}
Here,
\begin{equation}
  \alpha=\frac{1}{k} - \sqrt{ (1-1/k)(\gamma-1/k)}
\end{equation}
and $k$ is the integer such that $ \frac{1}{k} \le  \gamma \le \frac{1}{k-1}$.
\end{theorem}
\begin{proof}
Suppose we always have the solution in the form as
\begin{equation}\label{Eq:solutionPossible}
\begin{aligned}
  x_1=x_2=\dots=x_{k-1}, x_{k}, x_{k+1}=\dots x_d = 0.\nonumber
\end{aligned}
\end{equation}
Otherwise, there must exist three $x_i,x_j,x_k$ such that $x_i> x_j\ge x_k$ and $x_k\ne 0$. Following a similar argument in the proof of Theorem \ref{theorem:max}, we can show that this contradicts  Lemma \ref{lamme:min3}.

According to Eq.~\eqref{Eq:solutionPossible}, we have
\begin{equation}\label{Eq:}
\begin{aligned}
  (k-1)x_1^2 + x_k^2 &= \gamma,\nonumber\\
  (k-1)x_1+x_{k}&=1,\\
  k&\le d
\end{aligned}
\end{equation}
We can show that the possible integer value for $k$ is unique. That is,
\begin{equation}\label{Eq:}
\begin{aligned}
 k[(k-1)x_1^2 + x_k^2] &\ge [  (k-1)x_1+x_{k}]^2\nonumber\\
 &\ge (k-1)[(k-1)x_1^2 + x_k^2]
 \end{aligned}
\end{equation}
Equivalently, we have
\begin{equation}\label{Eq:}
\begin{aligned}
 k\gamma \ge 1  \ge (k-1)\gamma,\nonumber
 \end{aligned}
\end{equation}
hence
\begin{equation}\label{Eq:}
\begin{aligned}
\frac{1}{\gamma}\le &k  \le  \frac{1}{\gamma} + 1,\nonumber\\
 \frac{1}{k} \le  &\gamma \le \frac{1}{k-1}.
 \end{aligned}
\end{equation}

\end{proof}

\subsection{Upper and lower bounds to coherence and entanglement}
We now call $\{x_{i,\rho}^{\text{M}}\}, \{x_{i,\rho}^{\text{m}}\}$ the the vectors solving the maximization and the minimization, respectively.  By minimizing (maximizing) the coherence of the dephased state $\rho_d=\sum_i\ket{i}\!\!\bra{i}\rho\ket{i}\!\!\bra{i}$, and maximizing (minimizing) the coherence of the state under study, we obtain lower (upper) bounds to the relative entropy of coherence:\\
\emph{Result 1} --- The relative entropy of coherence $C_{\mathrm{RE}}(\rho)$ is  bounded as follows:
\begin{eqnarray}\label{Eq:Estimationlb}
l_c(\rho)\leq C_{\mathrm{RE}}(\rho)\leq u_c(\rho), 
\end{eqnarray}
{\footnotesize
\begin{eqnarray}
 &&l_c(\rho)= -(1-x^{\text{m}}_{k_{\rho_d},\rho_d}) \log x^{\text{m}}_{1,\rho_d}-x^{\text{m}}_{k_{\rho_d},\rho_d} \log x^{\text{m}}_{k_{\rho_d},\rho_d}\nonumber\\
  &+& \frac{(d-1)(1-x^{\text{M}}_{1,\rho})}{d} \log\frac{(1-x^{\text{M}}_{1,\rho})}{d} + x^{\text{M}}_{1,\rho} \log{x^{\text{M}}_{1,\rho}},\nonumber\\
 &&u_c(\gamma_{\rho,\rho_d})=  (1-x^{\text{m}}_{k_{\rho},\rho}) \log x^{\text{m}}_{1,\rho} +x^{\text{m}}_{k_{\rho},\rho} \log x^{\text{m}}_{k_{\rho},\rho}\nonumber\\
 &-&  \frac{(d-1)(1-x^{\text{M}}_{1,\rho_d})}{d} \log \frac{(1-x^{\text{M}}_{1,\rho_d})}{d} - x^{\text{M}}_{1,\rho_d} \log x^{\text{M}}_{1,\rho_d}\nonumber.
\end{eqnarray}
}
On the same hand, given a bipartite state $\rho_{AB}$, by minimizing (maximizing) the marginal purity on $B$ subsystem and maximizing (minimizing) the global purity, one has \\
\emph{Result 2.}--- Given a quantum state $\rho_{AB}\in \mathcal{H}_{d_A}\otimes \mathcal{H}_{d_B}$, and defining $\rho_B = \mathrm{Tr}_A\rho_{AB}$, its coherent information $I(A\rangle B)$ is bounded as follows: 
  \begin{eqnarray}  \label{thebound}
 l_e(\rho)\leq I(A\rangle B)\leq u_e(\rho),
 \end{eqnarray} 
{\footnotesize
 \begin{eqnarray} 
 &&l_e(\rho)= \nonumber\\
&-&(1-x^{\text{m}}_{k_{\rho_B},\rho_B}) \log x^{\text{m}}_{1,\rho_B} 
 -x^{\text{m}}_{k_{\rho_B},\rho_B} \log  x^{\text{m}}_{k_{\rho_B},\rho_B}\nonumber\\
  &+& \frac{(d-1)(1-x^{\text{M}}_{1,\rho_{AB}})}{d} \log \frac{(1-x^{\text{M}}_{1,\rho_{AB}})}{d} + x^{\text{M}}_{1,\rho_{AB}} \log {x^{\text{M}}_{1,\rho_{AB}}},\nonumber\\
&&u_c(\gamma_{\rho_{AB},\rho_B})=\nonumber\\
&& (1-x^{\text{m}}_{k_{\rho_{AB}},\rho_{AB}}) \log x^{\text{m}}_{1,\rho_{AB}} +x^{\text{m}}_{k_{\rho_{AB}},\rho_{AB}} \log  x^{\text{m}}_{k_{\rho_{AB}},\rho_{AB}}\nonumber\\
 &-&  \frac{(d-1)(1-x^{\text{M}}_{1,\rho_B})}{d} \log \frac{(1-x^{\text{M}}_{1,\rho_B})}{d} - x^{\text{M}}_{1,\rho_B} \log x^{\text{M}}_{1,\rho_B}\nonumber.
\end{eqnarray}
}

\end{document}